\documentclass[runningheads]{llncs}
\usepackage[T1]{fontenc}
\usepackage{graphicx}
\usepackage{bbm}
\usepackage{amsmath,amssymb}
\usepackage{xcolor,tikz,url}
\usepackage{soul,booktabs}

\usepackage{times}  

\newcommand{\TM}[2]{#1^{(#2)}}

\usepackage[colorlinks=true,linkcolor=magenta,citecolor=magenta]{hyperref}
\usepackage{amsthm}
\newtheorem*{theorem*}{Theorem}
\newtheorem*{T1}{Theorem~\ref{thm1}} 
\begin{document}
\title{Serial Monopoly on Blockchains with Quasi-patient Users}

\author{Paolo Penna\inst{1} \and Manvir Schneider\inst{2}}
\institute{IOG \and Cardano Foundation}
\maketitle              
\begin{abstract}
In the face of limited block size, miners (e.g., in Bitcoin) typically prioritize transactions with the highest bids, which increasingly make up a larger portion of their revenue. If the block size were to expand significantly, meeting all transaction demand due to infrastructure or protocol improvements, bids could drop to zero or to a constant minimum fee. This would diminish miners' incentives to mine, potentially affecting network security. To address this, Lavi et al.~\cite{Lavi2022} introduced a monopolistic pricing mechanism where miners may not fill the entire block but only include transactions that pay a minimum price set by the miner. This mechanism aims to be incentive-compatible and allows miners to collect some revenue, although it may result in an unbounded loss in welfare. Nisan~\cite{Nisan} expands this by modeling bidders as \textit{patient}, meaning they are willing to wait without cost until block prices drop low enough for their transactions to be included, leading to wildly fluctuating prices even when demand is stable and there is no stochastic element in the model.
In order to capture users' diminishing interest in having their transactions added to the ledger over time, we consider a more realistic setting with \textit{quasi-patient} users, where only a fraction $\delta \in [0,1]$ of pending transactions remains in the next round. This richer model encompasses both Lavi et al.'s \cite{Lavi2022} \emph{impatient} users ($\delta=0$) and Nisan's \cite{Nisan} \emph{patient} users ($\delta=1$) as special cases. We demonstrate that Nisan's fluctuating dynamics persist for $\delta$ close to 1, while for $\delta$ close to 0, the dynamics resemble the impatient case. For $\delta \in (0,1)$, we establish new bounds on price dynamics, revealing unexpected effects. Unlike the fully patient case, the bounds of the dynamics for $\delta<1$ depend on the demand curve and undergo a ``transition phase''. For some $\delta$, the model mirrors the fully patient setting, and for smaller $\delta' < \delta$, it stabilizes at the highest \textit{monopolist} price, thus collapsing to the impatient case. We provide quantitative bounds and analytical results, showing that the bounds for $\delta=1$ are generally not tight for $\delta<1$, and we give guarantees on the minimum (``admission'') price for transactions.

\end{abstract}

\section{Introduction}
Transaction fee mechanisms are a fundamental part of a blockchain. 
A block leader, in general, has full freedom to choose which transaction from the public mempool (and private mempool) to include in a block. A well designed transaction fee mechanism contributes to maximizing the social welfare, which is the total value of the chosen transactions subject to the block size constraint.\footnote{The definition social welfare may also include the value of the block producers, see Bahrani et al.~\cite{Bahrani2024}. However, how to derive or estimate these values, both for users and block producers is a difficult question.} 
The most prominent proof-of-work blockchain, Bitcoin~\cite{Nakamoto2008}, employs a pay-your-bid mechanism. In particular, the higher the bid attached to a transaction, the higher the chances to be included in the next block by the miner. The user-paid bids are rewarded to the miner. Myopic rational miners will always try to maximize their revenue and will therefore choose the highest paying transactions. Since block space is scarce, this mechanism can drive prices high, especially when there is congestion and hence transaction inclusion might experience a considerable delay. 
On the contrary, if block space becomes large enough to meet all demand, the fees/bids would drop to 0 and miners are left with no incentive to mine and thus will compromise security.

Opposing to the pay-your-bid mechanism is the dynamic posted price mechanism, like EIP-1559. In Ethereum's EIP-1559, the transaction fee is split into base fee and a tip. The block proposer only receives the tips while the base fee is burned. Rational block proposers will therefore only select transactions that pay at least the base fee.\footnote{It may be that a block proposer has some positive intrinsic value for some transaction that pays less than the base fee, and therefore includes this transaction by paying the remaining base fee himself. Note that this is related to active block proposers, see Bahrani et al.~\cite{Bahrani2024}.} The target of EIP-1559 is to have half-full blocks. If the previous block was filled less (more) than the target, the base fee is lowered (increased) accordingly. The dynamic base fee allows to handle low and high demand phases and satisfies various good properties of transaction fee mechanisms such as incentive-compatibility properties and off-chain collusion proofness \cite{Roughgarden2021}.\footnote{Note that also Bitcoin's fee mechanism and Ethereum's pre-EIP-1559 fee mechanism allow to handle low and high demand phases but do not satisfy various good properties of transaction fee mechanisms.}  

A different approach to transaction fee is that of Cardano, where there is a constant minimum fee and a fixed fee per byte. The transaction selection process is a first-in-first-out (FIFO) mechanism. The transaction fees of included transactions are collected in a `pot' and distributed at the end of an epoch (5 days) \cite{cardano-monetary}.
The rewards (fees) are distributed to the block producers (stake pools) proportionally to the number of proposed blocks in the epoch.\footnote{For the Cardano blockchain, each stake pool, given their stake, has an expected number of blocks that it will produce during an epoch. The rewards that a stake pools receives at the end of an epoch are scaled by a performance variable which is the fraction of the actual number of produced blocks and the expected number.} 

In the event that the (minimum) fees are not high enough to incentivice block producers (to produce blocks) and thus security is at risk, a question one could ask is whether block leaders should be allowed to set their own fees rather than a fee imposed by the protocol, and if so, whether the fee should be the same for all and how it would affect blockchain security. The study of such a monopolistic pricing mechanism is part of Nisan~\cite{Nisan} and this paper. In the monopolist pricing mechanism, transactions willing to pay at least the price set by the monopolist are included in a block (until the block is full). Unlike in the pay-your-bid mechanism, all included transactions pay exactly the price set by the monopolist (rather than their bid). Or, in the words of Lavi et al.~\cite{Lavi2022}, the monopolist chooses the number of accepted transactions in the block and all transactions pay the smallest bid among the accepted transactions.

In his paper~\cite{Nisan}, Nisan assumes that block leaders set their own prices and all transaction that are not included in a block remain in the mempool forever until they are picked up eventually in a future block. On the one hand, this better captures the role and importance of the mempool in  blockchain price dynamics, revealing the following surprising effect:
\begin{quote}
    \emph{``[...] prices keep fluctuating wildly and this is an endogenous property of
the model and happens even when demand is stable with nothing stochastic in
the model.'' (\cite{Nisan})}
\end{quote}
On the other hand, the assumption above, is too strong and does not fully reflect real world behavior. In fact, partially impatient users may cancel their transaction after some time if not included in a block. The Cardano blockchain supports expiring transactions, allowing users to set an expiry time (or validity interval) for their transactions. If a transaction is not recorded in the ledger before its expiry time, it is discarded. This feature incurs no additional cost for users.\footnote{
On other blockchains, users cannot always cancel transactions without incurring a fee, as this could lead to denial-of-service attack vulnerabilities. Blockchains like Bitcoin and Ethereum mitigate this by allowing miners and validators to set a minimum fee increase percentage for users who wish to replace their transactions retroactively. For instance, some Bitcoin wallets support the Replace-by-Fee protocol, enabling users to replace an existing transaction with one that pays a higher fee.} Therefore, we  assume that only a fraction of unsupplied transactions remains in the next round, while the other fraction of unsupplied demand is removed (from the mempool).

\begin{figure}[t]
    \centering
    \includegraphics[width=\linewidth]{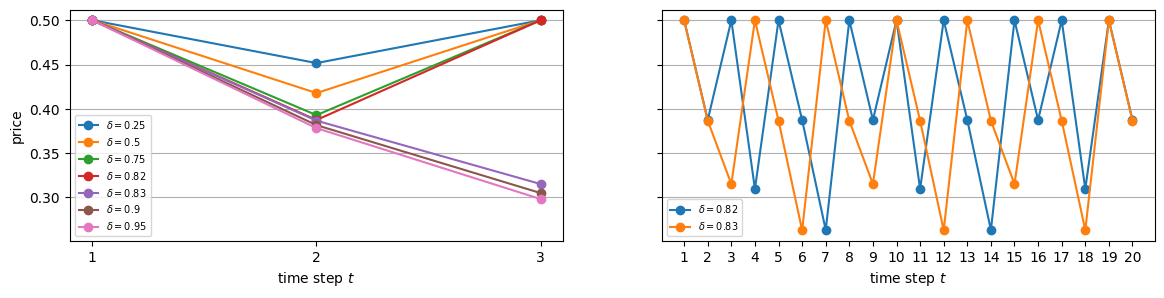}
    \caption{Price dynamics from Example~\ref{intro:example} for different values of $\delta$, with 3 time steps (left) and 20 time steps (right), respectively. For $\delta$ smaller than $2\sqrt{2}-2\approx 0.828$ the prices jumps up after step $t=2$. For $\delta$ above the threshold, the price decreases after step $t=2$.}
    \label{fig:example1}
\end{figure}

\subsection{Our Contributions}
We put forward a model for monopolist pricing dynamics tailored to accommodate \emph{quasi-patient} users (see Section~\ref{sec:model} for details and formal definitions). Our model incorporates  a ``decay'' parameter $\delta\in [0,1]$ which corresponds to the fraction of pending transactions that remain in the mempool at the next round. Thus, a fraction $1-\delta$ of pending transactions gets withdrawn by the users at each round. This can be motivated by users' diminishing interest in having their transactions included in a block, due to factors such as waiting costs. The case of \emph{impatient users} in Lavi et al.~\cite{Lavi2022} corresponds to $\delta=0$, and the case of \emph{patient users} in Nisan~\cite{Nisan} to $\delta=1$. Our model spans all intermediate scenarios between these two extreme cases, and it allows us to study how monopolist pricing mechanisms behave at different \emph{patience} levels $\delta$. In particular, we no longer assume \emph{patient users} who are willing to wait (and pay) indefinitely long for their transactions to be included on the ledger. 
A simulation of few steps of the resulting dynamics is shown in Figure~\ref{fig:example1} for several values of  $\delta$ and a simple demand function. At each time step, only transactions bidding at least the price in the graph are included, and thus the following questions arise:
\begin{quote}
    \emph{Given some daily demand function, what is the (minimum) price that users must pay to have their transactions  eventually included? What is the price that users have to pay to have their transactions included immediately? Is it even possible to analytically compute these prices  and how does $\delta$ affect them for arbitrary demand functions?}
\end{quote}
In Figure~\ref{fig:example1} (left), we observe two distinct behaviors in the price dynamics. For some values of $\delta$, the prices jump up after step 2, while for others, the prices continue to decrease. Another interesting question arises: will the dynamics for smaller $\delta$ eventually reach the same prices as those for larger $\delta$? For example, does the blue graph in Figure~\ref{fig:example1} (left) eventually reach the low prices observed in the purple graph?

We provide analytical results on the monopolistic pricing dynamics for any $\delta\in(0,1)$.  Our findings highlight that monopolistic price mechanisms for \emph{quasi-patient} users still posses good features, though with some key differences with the case of \emph{patient} users. 
In Sections~\ref{sec:delta-large} and~\ref{sec:lower-bounds}, we analyze the dynamics for different values of $\delta$ and how it affects their behavior. In particular, we demonstrate that in regimes with a \emph{sufficiently small fraction of expiring transactions} ($\delta<1$ sufficiently large), the dynamics behaves qualitatively similarly to the case of \emph{no expiring transactions} ($\delta=1$). The analysis highlights several differences, particularly how the ``structure'' of the demand function  $Q$ influences the dynamics for $\delta < 1$ compared to the case $\delta=1$.

Specifically, Theorem~\ref{thm1} informally states that:
\begin{itemize}
\item \emph{Prices decrease or jump up to maximum price.} The price for being included in the current block is either smaller than the one at the previous block, or it is the maximum price, which is the so called monopolist price (the price that is always asked if users are \emph{fully-impatient} and there is no pent-up demand \cite{Lavi2022}). 
\item \emph{Immediate inclusion price (monopolist price).} The largest price that the monopolist ever asks (immediate transaction inclusion guaranteed) equals the monopolist price.  Also for quasi-patient users the dynamics ask the monopolistic (maximum) price infinitely often, meaning that immediate inclusion is \emph{not} guaranteed for any lower price. 
     
    \item  \emph{Minimum admission price.} We call the lowest price the monopolist ever sets the \emph{minimum admission price}, and it is bounded based on $\delta$. If there is sufficient pent-up demand, the minimum admission price becomes low, allowing transactions to eventually be included at a cheaper price.
\end{itemize}

A direct comparison between our bounds for quasi-patient users and the case of patient users, shows the following. First,  our upper bounds on the minimum admission price depend on $\delta$. Second, the minimum admission price for $\delta<1$ is at least the minimum admission price for $\delta=1$ (cf. our Theorem~\ref{thm1} and Theorem~\ref{th:Nisan} below from \cite{Nisan}). In Section~\ref{sec:lower-bounds}, we prove lower bounds on the minimum admission price. In particular, Theorem~\ref{prop:Q-eps-fun} states the following:

\begin{itemize}
    \item \emph{The minimum admission price for impatient users is never tight for quasi-patient users}. That is, for every $\delta<1$ there is a demand function for  which the minimum admission price is strictly higher.  
    \item \emph{Collapse to the impatient case.} The change in the dynamics is not continuous in $\delta$. For some small $\delta>0$, the dynamics behave exactly like the case of \emph{impatient users} ($\delta=0$). That is, the minimum admission price coincides with the (maximum) monopolist price at all time steps.  
\end{itemize}
Furthermore, the above mentioned collapse means that the positive effect of pent-up demand, which results in a minimum admission price smaller than the monopolist price, may completely be nullified for quasi-patient users.
Another important difference is that, for patient users, the minimum admission price is ``almost'' independent of the structure of the demand function (it only depends on the block size $s$  and on the revenue at the monopolist price). This is no longer true for quasi-patient users, where the ``overall structure'' of the demand function seems to play a role.  
Without further assumptions on the demand function, the  conditions under which  Theorem~\ref{thm1} can be applied ($\delta>\bar{\delta}_{ser}$) are essentially tight as the collapse already happens slightly below $\bar{\delta}_{ser}$ -- see Theorem~\ref{thm:tightness-main}.

\subsection{Related Work}
Transaction fee mechanisms are analyzed from the perspective of mechanism design in Roughgarden~\cite{Roughgarden2021}. Additionally the dynamic posted price mechanism EIP-1559 is analyzed. Follow up work on transaction fee mechanism design includes \cite{Ferreira2021,Chung2023,DynamicTFMD,Chen2024,Bahrani2024}. More work focusing on the dynamics of EIP-1559 (and potentially chaotic behavior) is \cite{roughgarden2020transaction,leonardos2021,reijsbergen2021,leonardos2023}. 

\paragraph*{Monopolistic pricing mechanisms.} The monopolistic pricing mechanism was initially examined by Goldberg et al.~\cite{Goldberg2006}. Subsequently it has been analyzed within the context of blockchains by Lavi et al. \cite{Lavi2022}, Yao~\cite{Yao2018} and Basu et al.~\cite{Basu2019} prior to Nisan~\cite{Nisan} and this paper. In particular, it is motivated by its ability to maintain fee revenue when block rewards are low and blocks can be large in size.
Lavi et al.~\cite{Lavi2022} study the monopolistic pricing mechanism and describe the mechanism as follows: (1) Transactions specify bids (maximal fee) they are willing to pay; (2) Miners (or block leaders/monopolists) choose which subset to include in their block; (3) All transactions in the block pay the exact same fee which is equal to the smallest bid among the included transactions; (4) Miners maximize their revenue which is the product of the minimal bid and the number of included transactions. The focus of their paper is on a single shot game where users are maximally impatient in the sense that they derive utility from immediate block inclusion and no utility for inclusion in a future block. They show that truthful bidding (users bidding their true valuation) is ``nearly'' an equilibrium, i.e. relative gains from strategic bidding go to zero as number of transactions increase. The monopolistic pricing mechanism collects at least as much revenue from maximally impatient users as the pay-your-bid mechanism (as employed in Bitcoin).
Yao~\cite{Yao2018} builds on the work of \cite{Lavi2022} and studies properties of the monopolistic pricing mechanism, in particular, incentive compatibility when users' valuations are drawn from an i.i.d. distribution. Basu et al.~\cite{Basu2019} study a setting similar to \cite{Lavi2022}, where the miner has to fill the block up to a certain level to receive the reward, which consists of mean of the fees from the last $B$ blocks, including the current constructed block (note some similarity to Cardano \cite{cardano-monetary} and to Eyal et al.~\cite{Eyal16}). The model of Eyal et al.~\cite{Eyal16} assumes many miners, incentivizing them to pick the highest-value transactions, and aims to maximize social welfare. Note that the model of \cite{Lavi2022} does not aim to maximize social welfare. To see this, note that, if the monopolist chooses a subset of transactions that does not fill the block entirely, the monopolist could potentially include transactions with lower bids. However, doing so would decrease the price that all included transactions have to pay and hence would lower the monopolist's revenue.

Nisan~\cite{Nisan} studies a monopolist pricing mechanism, in which each block leader (or proposer) is allowed to choose the price $p$ for his block. Transactions willing to pay at least $p$ may be included by the monopolist and all included transactions pay exactly $p$.\footnote{In principle, this mechanism is the same as in Lavi et al.~\cite{Lavi2022}. In \cite{Lavi2022} the fee to be paid by users is determined by the lowest bid $p$ of the included transactions. There is at least one transaction with bid $p$ (which is the lowest bidding transaction), while in Nisan~\cite{Nisan} there need not be a transaction with bid exactly equal to $p$.} Rationality of block leaders implies that the block leaders will choose a price that maximizes their revenue given price and the block space filled by the chosen transactions.\footnote{While the other mechanism of Bitcoin and Ethereum maximize the total value of included transactions subject to the available block space (i.e. social welfare), the monopolist mechanism maximizes the block leaders revenue.} Nisan's model involves infinitely patient users, i.e. users' valuations of transaction inclusion do not depreciate over time. Transactions stay in the mempool until eventually picked up by some block leader.
Block leaders face the same demand distribution at every step in time plus the pent-up demand from the previous steps, that is, additionally to the daily demand the block leaders faces the transactions that were not picked up by previous block leaders. When optimizing given the current total demand, the block leader only optimizes for the current block (myopic block leader). Furthermore, the available block space for each block is fixed and demand is known to the block leader. 

Kiayias et al.~\cite{Tiered} study a mechanism to account for transactions with different priority/urgency. In particular, the mechanism splits blocks into different tiers with each tier having its own characteristics such as fee and size. The fee and size are dynamically adjusted based on previous demand and fees. This mechanism ensures that high priority transactions can choose to be included in a tier with high priority by paying high transaction fee.

Patient blockchain users are analyzed in Huberman et al.~\cite{Huberman21}, Gafni~and~Yaish~\cite{Gafni2022}.
The study of non-myopic miners includes expiring transactions \cite{Gafni2024}, base-fee manipulations \cite{Azouvi2023}, under-cutting attacks \cite{Undercutting1}.

\section{Model}\label{sec:model}
We extend the model of Nisan~\cite{Nisan} for non-strategic agents with an additional parameter $\delta\in [0,1]$ which corresponds to the fraction of pending transactions remaining  at next round (thus, $1-\delta$ is the fraction of pending transactions withdrawn\footnote{Or in other words, a fraction $1-\delta$ of pending transactions expires at each time step (cf. Cardano blockchain for expiring transactions)} by the corresponding users -- see below).
 The dynamics is specified as follows: 
\begin{itemize}
    \item \emph{Time} is discrete and indexed by $t = 1,2,\ldots,$.
    \item \emph{Daily demand}: A demand function $Q$ quantifies the daily demand $Q(p)$ for every price level $p$. Function $Q$ is continuous and decreasing in $p$ as $Q(p)$ is the number of newly added transactions willing to pay $p$ or more to be included.
    \item \emph{Monopolist}: A monopolist (chosen for the current round $t$) faces a total demand $D_t$ consisting of daily demand and pent-up demand from previous rounds (see below). As $D_t(p)$ is the total number of transactions willing to pay at least $p$, the monopolist chooses a price maximizing his own revenue subject to the supply constraint $s$ (block size = max number of transactions per block): 
    \begin{align}
        p_t = \arg \max_p p\cdot \min(s, D_t(p)) \ . 
    \end{align}
    The corresponding supplied quantity is $q_t = D_t(p_t)$, and the monopolist's revenue (at time $t$) is 
    $\texttt{REV}_t := p_t \cdot q_t$.
    \item \emph{Pent-up demand}: Initially there is no pent-up demand from previous rounds, that is,  $Z_0(p)=0$ for all $p$. The pent-up demand at time $t\geq 1$ is 
\begin{align}\label{eq:Z-def}
Z_t(p) :=
    \begin{cases}
        D_t(p)- q_t & \text{for } p\leq p_t \\
        0 & \text{for } p> p_t
    \end{cases} \ . 
\end{align}
    \item \emph{Total demand and $\delta$}: Only a fraction $\delta\in [0,1]$ of pent-up demand survives to the next round, and thus total demand is
\begin{align}\label{eq:D-def}
    D_t(p)= \delta \cdot Z_{t-1}(p) + Q(p)\ . 
\end{align}  
\end{itemize}

\begin{remark}
    For $\delta=1$ the model above boils down to the one in \cite{Nisan} where all transactions not included in the current round remain in the system and they are eventually included if an only if their price is above some minimum price $p_{ser}$. For $\delta=0$  transactions are either immediately included or they disappear, thus implying that the dynamics above stay at the monopolist price $p_{mon}>p_{ser}$ and only transactions willing to pay this price are included. 
\end{remark}

In the sequel we shall focus on the case $\delta\in (0,1)$ as the case $\delta=0$ is trivial and $\delta=1$ coincides with the model in \cite{Nisan}.

\paragraph*{Key quantities.}
Note that by the definition of the total demand and pent-up demand we can write the total demand as follows.
\begin{remark}
The total demand at time $t$ can be rewritten as 
\begin{align}\label{eq:D-rewritten}
    D_t(p) = 
    \begin{cases}
        a_t \cdot Q(p) - b_t  &  \text{for } p \leq \displaystyle \min_{1\leq \tau \leq t-1}p_\tau  \\ 
        Q(p) & \text{otherwise}
    \end{cases}
\end{align}
    where from \eqref{eq:Z-def} and \eqref{eq:D-def} we have
\begin{align}\label{eq:at-bt-quantities}
    a_t =  1 + \delta+\cdots + \delta^{t-1} \ , && b_t =  q_1 \delta^{t-1} + q_2\delta^{t-2}+\cdots + q_{t-1}\delta \ .
\end{align}
Note that $a_t = \frac{1-\delta^t}{1-\delta} $ for $\delta\in (0,1)$, and $a_t = t$ for $\delta = 1$. 

\end{remark}

As one of our main results (see Section~\ref{sec:delta-large}) show,  the price dynamics fluctuate between two prices that involve the following quantities:
\begin{definition}
For any demand function $Q$ and any supply $s$, the corresponding \emph{monopolist price} $p_{mon}$ and \emph{serial price} $p_{ser}$ are defined as follows: \begin{align}
    p_{mon}:= \arg \max_{p} p \cdot \min(s,Q(p)) \  ,  && q_{mon} := Q(p_{mon}) \  , 
    \\
    p_{ser}:= p_{mon}\cdot q_{mon}/s \  , && q_{ser} := Q(p_{ser}) \  . 
\end{align}
\end{definition}

Note that the \textit{monopolist price} $p_{mon}$ is simply the price that maximizes the revenue of the monopolist when facing demand $Q(p)$. 

\begin{remark}
    Since $D_t(p)\geq Q(p)$ at any time $t\geq 1$, the monopolist can always obtain the revenue at the monopolist price $\texttt{REV}_{mon} := p_{mon} \cdot q_{mon}$ by choosing price $p_{mon}$. Therefore, we have  $\texttt{REV}_t = p_t \cdot q_t \geq \texttt{REV}_{mon}$  for all $t$. 
\end{remark}

It turns out that these prices characterize tightly the dynamics for the case of \emph{patient} users ($\delta=1)$, as shown  by the next definition and theorem. 

\begin{definition} [Eventual Transaction Inclusion, (Minimum) Admission Price] \label{def:eventualinclusion,admission,minAdmission}
For a given price dynamic we consider the following definitions:
\begin{itemize}
    \item A transaction with price $p$ is \textit{eventually included} if there exists $\Delta_p$ such that, for every $T\geq 1$, there exists some $t$ with $p_t\leq p$ and $T \leq t \leq T+\Delta_p$. 
    \item  A price $p$ is called \textit{admission price} if all transactions paying $p$ are eventually included.
    \item The minimum admission price $p_{map}$ is the smallest admission price such that all transactions paying at least $p_{map}$ are eventually included. 
\end{itemize}
\end{definition}

Next, we state the main result from Nisan~\cite{Nisan}.
\begin{theorem}[Theorem~1 in \cite{Nisan} restated]\label{th:Nisan}
 For patient users ($\delta=1$) and for any strictly decreasing demand function $Q$ and supply $s$ the following holds:
    \begin{enumerate}
        \item The dynamics stay always between $p_{ser}$ and $p_{mon}$, that is,  prices $p_t$ satisfy $p_{ser}\leq p_t \leq p_{mon}$ for all $t\geq 1$.  In particular, transactions paying less than $p_{ser}$ will never be included. 
        At each step $t$, the prices either decrease ($p_t < p_{t-1}$) or they jump up to the monopolist price ($p_t=p_{mon}$). 
        
        \item Every price  larger than $p_{ser}$ is an admission price. Therefore, $p_{ser}$ is the minimum admission price ($p_{map}=p_{ser}$).
        Moreover, the dynamics  pass through the monopolist price $p_{mon}$ infinitely often. 
    \end{enumerate}
    Transactions paying at least $p_{mon}$ are immediately included, and this is tight as there are infinitely steps for which paying less will delay admission to a later step. 
\end{theorem}

We now go back to the case of \emph{quasi-patient} users ($\delta\in(0,1)$). In order to illustrate the dynamics and, in particular, how it differs from the case $\delta=1$, we next revise and adapt an example in \cite[Section~2.2]{Nisan}. 

\begin{example}\label{intro:example}
Let $Q(p)=1-p$ for $p\in [0,1]$, $s=1$ and let $\delta \in (0,1)$. Then, we have the following:
\begin{description}
    \item[($t=1$)] The initial demand is $D_1(p)=Q(p)$ and thus we maximize $pD_1(p)=p(1-p)$ which gives $p_1 = 0.5$ as maximizer and $q_1 = D_1(p_1) = 0.5$. This price $p_1$ is the monopolist price and the revenue is $\texttt{REV}_1 = \texttt{REV}_{mon} = p_1 q_1 = 0.25$. The pent-up demand is $Z_1(p) = D_1(p)- q_1 = \frac{1}{2}-p$ if $p<p_1$ and zero otherwise.
    
    \item[($t=2$)] The total demand is $D_2(p) = \delta Z_1(p) +Q(p) = 1+ \delta/2 - (1+\delta)p$ if $p<p_1$ and $Q(p)$ otherwise. We maximize $pD_2(p)$ and get $p_2 = \frac{2+\delta}{4(1+\delta)}$ as maximizer and thus $q_2 = \frac{2+\delta}{4}$. Note that $p_2 < p_1$ if and only if $\delta>0$. The revenue is $\texttt{REV}_2 = \frac{(2+\delta)^2}{16(1+\delta)}$. Note that $\texttt{REV}_2 > \texttt{REV}_{1}$ if and only if $\delta > -1$, that is, the revenue from step $t=2$ is better than the monopolist revenue.
    The pent-up demand is $Z_2(p) = D_2(p)-q_2 = (2+\delta)/4 - (1+\delta)p$ if $p<p_2$ and zero otherwise.
    
    \item[($t=3$)]  The total demand is $D_3(p) = \delta Z_2(p) +Q(p) = \frac{2\delta + \delta^2 + 4}{4}- (1+\delta + \delta^2)p$ if $p<p_2$ and $Q(p)$ otherwise. Maximization yields $p_3 = \frac{2\delta + \delta^2 + 4}{8(1+\delta + \delta^2)}$. Note that $p_3 < p_2$ if and only if $\delta>0$.
    Hence, $q_3 =\frac{2\delta + \delta^2 + 4}{8}$ and the revenue is $\texttt{REV}_3 = \frac{(2\delta + \delta^2+4)^2}{64 (1+\delta + \delta^2)}$. Note that $\texttt{REV}_3 > \texttt{REV}_1 $ if and only if $\delta>\delta^\star : = 2\sqrt{2}-2\approx 0.828$. That is, if $\delta > \delta^\star$, the price dynamics do not jump and take $p_3$ as above. The pent-up demand is $Z_3(p) = \frac{2\delta + \delta^2 + 4}{8} - (1+\delta + \delta^2)p$ if $p<p_3$ and zero otherwise. However, if $\delta \leq \delta^\star$, the revenue will be less than the monopolist revenue. In this case, the price dynamics would jump up to the monopolist price and take $p_3 = p_{mon} = 0.5$ and thus $q_3 = q_{mon}=0.5$. The pent-up demand would be $Z_3(p) = D_3(p)-q_3 = Q(p) - 0.5 = 0.5-p$ if $p<p_3$ and zero otherwise.
    \end{description}
    The price dynamics for the first 3 steps and  20 steps are depicted in Figure~\ref{fig:example1} for different values of $\delta$.  Figure~\ref{fig:line-delta-0.5} shows the dynamics for a longer period of 100 steps.
\end{example}

\section{Upper Bounds on the Minimum Admission Price}\label{sec:delta-large}
In this section, we provide a bound on the minimum price which guarantees transactions to be \emph{eventually} included, depending on $\delta$. The main result is summarized by the following definition and  theorem (Definition~\ref{def:adapted-quantities}
and Theorem~\ref{thm1}).

\begin{definition}\label{def:adapted-quantities}
    For any continuous decreasing demand function $Q$ and supply $s$ we define the following quantities:
    \begin{align}\label{eq:def-pserbar-qserbar}
        \overline{p}_{ser}:= \frac{p_{ser} \cdot s}{q_{ser}}\ ,  && \overline{q}_{ser}:= Q(\overline{p}_{ser})  \ ,  && \overline{\delta}_{ser} := 1-\frac{q_{ser}-\overline{q}_{ser}}{s} \ .
    \end{align}
    Moreover, for any $\delta > \overline{\delta}_{ser}$, we let $\TM{p_{ser}}{\delta}$ be the price such that~\footnote{This price exists by continuity and monotonicity of $Q$, and because $Q(\overline{p}_{ser}) = \overline{q}_{ser} = q_{ser} - (1-\overline{\delta}_{ser})\cdot s<q_{ser} - (1-\delta)\cdot s\leq q_{ser} = Q(p_{ser})$, where last inequality is due to $\delta \leq 1$.}
\begin{align}\label{eq:def-pserdelta}
    Q(\TM{p_{ser}}{\delta}) = q_{ser} - (1-\delta)\cdot s \ . 
\end{align}
\end{definition}

\begin{example}[Example~\ref{intro:example} continued] For the setting in Example~\ref{intro:example} we observe the minimum admission prices $p_{map}$ (over 100 steps) for every $\delta \in [0,1]$ and display it in Figure~\ref{fig:pmap}.

\begin{figure}[t]
    \centering
    \includegraphics[width=0.6\linewidth]{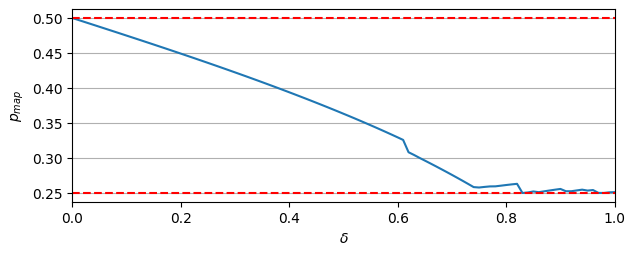}
    \caption{Minimum admission price $p_{map}$ (derived from 100 steps) depending on $\delta$ for daily demand $Q(p)=1-p$ and supply $s=1$. Fluctuations for larger $\delta$ may be due to a lack of steps for the price dynamics to reach its minimum.}
    \label{fig:pmap}
\end{figure}
    
\end{example}

We next state the main result of this section, which highlights the main similarities and differences with the case of impatient users (see Theorem~\ref{th:Nisan} above). 

\begin{theorem}\label{thm1}
   For any strictly decreasing demand function $Q$ and supply $s$ the following holds:
    \begin{enumerate}
        \item \label{itm:thm1:item1} The minimum admission price is at least $p_{ser}$, and thus transactions paying less than this price will never be included. In particular, the dynamics stay always between $p_{ser}$ and $p_{mon}$, that is,  prices $p_t$ satisfy $p_{ser}\leq p_t \leq p_{mon}$ for all $t\geq 1$. 
        Moreover, at each step $t$, the prices either decrease ($p_t < p_{t-1}$) or they jump up to the monopolist price ($p_t=p_{mon}$). 
        \item \label{itm:thm1:finite-delta}
        For every  $\delta>\overline{\delta}_{ser}$, the minimum admission price is at most $\TM{p_{ser}}{\delta}$ defined by \eqref{eq:def-pserdelta} which satisfies $p_{ser}< \TM{p_{ser}}{\delta} < \overline{p}_{ser}$.
        Moreover, the dynamics  pass through the monopolist price $p_{mon}$ infinitely often. 
        \item \label{itm:thm1:item3} Every price  larger than $p_{ser}$ is an admission price for a sufficiently large $\delta$. 
        That is, for every $p^\star >p_{ser}$, there exists $\delta_{min}(p^\star)<1$ such that $p^\star$ is an admission price for every $\delta > \delta_{min}(p^\star)$. 
        Moreover, the dynamics  pass through the monopolist price $p_{mon}$ infinitely often. 
    \end{enumerate}
    Therefore, transactions paying at least $p_{mon}$ are immediately included, and this is tight to guarantee immediate inclusion (as there are infinitely many steps for which paying less will delay admission to a later step). 
\end{theorem}

\begin{proof}[Proof Idea] The proof is deferred to Appendix~\ref{app:thm1}. Here we discuss only the main intuition.   
    According to \eqref{eq:D-rewritten}-\eqref{eq:at-bt-quantities}, the pent-up demand cannot increase arbitrarily in our setting ($\delta<1$). We thus need to identify the pairs $\delta$ and $p^*$ for which, at price $p^*$, the pent-up demand becomes ``sufficiently high'' to induce the dynamics to take this or a smaller price. Moreover, we have to show that the dynamics cannot simply keep decreasing the prices, but must eventually jump back to the monopolist price, and then again reach $p^*$ or a smaller price as before infinitely often. The quantities in Definition~\ref{def:adapted-quantities} can be used to identify such pairs $\delta$ and $p^*$, which intuitively correspond to $p^*$ being ``not too small'' (Item~\ref{itm:thm1:finite-delta}) or $\delta$ being ``large enough'' (Item~\ref{itm:thm1:item3}). 
\end{proof}

The above result allows us to compare the case of quasi-patient users ($\delta <1$) with the case of patient users ($\delta=1$) analyzed in Theorem~\ref{th:Nisan}. Item~\ref{itm:thm1:item1} in Theorem~\ref{thm1} says that users need to pay at least $p_{ser}$ in order to have their transactions admitted, though this may still not be enough depending on $\delta$. How much they have to pay in order to guarantee inclusion is addressed in Items~\ref{itm:thm1:finite-delta} and~\ref{itm:thm1:item3}.
Item~\ref{itm:thm1:finite-delta} in Theorem~\ref{thm1} provides an upper bound on the minimum admission price, provided $\delta$ being large enough (condition $\delta > \overline{\delta}_{ser}$). This condition on $\delta$ is necessary as implied by the results in the next section, where we prove  lower bounds on the minimum admission price. Item~\ref{itm:thm1:item3} shows that the minimum admission price tends to $p_{ser}$ from above for $\delta \rightarrow 1$. Finally, Theorem~\ref{thm1} implies that the price dynamic does not converge to some price $p'<p_{mon}$. This is not obvious, as in our setting the pent-up demand does not grow arbitrarily and thus, in principle, it might be possible to have a sequence of decreasing prices.

\section{Lower Bounds on the Minimum Admission Price}\label{sec:lower-bounds}
In this section, we complement the results in the previous section, by showing that transactions below a certain price will  \emph{never} be included, depending on $\delta$.

Recall that the monopolist aims to choose a price $p$ maximizing the revenue $\texttt{REV}_t(p) := p\cdot \min(s,D_t(p))$ at the current step $t$. At every time step, there is always the option to choose the monopolist price $p_{mon}$ and receive revenue $\texttt{REV}_{mon}$. We can  compare $\texttt{REV}_{t}(p)$ and $\texttt{REV}_{mon}$ at any $t$ by considering the following function:
\begin{equation}\label{revenue-diff}
    f_t(p) := p\cdot D_t(p) - p_{mon}q_{mon}\ .
\end{equation}
Since the revenue for price $p$ satisfies $\texttt{REV}_t(p) \leq p\cdot D_t(p)$, if the function above is negative for some $p$, it means that the revenue at $p$ is worse than $\texttt{REV}_{mon}$, and therefore the next price $p_t$ cannot be $p$.
Observe that evaluating $D_t(p)$ and thus $f_t(p)$ is rather complex because of the ``previous history'' component involving $q_{t-1},\ldots,q_1$ -- see Equations~\eqref{eq:D-rewritten} and \eqref{eq:at-bt-quantities}. 
We next provide a simpler function  to evaluate for a generic $Q$, which still can be used to determine ``forbidden'' prices for the dynamics: 
\begin{align}\label{eq:lb-improved-F}
        F_t(p) := p \cdot   (a_t\cdot Q(p) - (a_t-1) q_{mon}) - p_{mon}q_{mon} \ , && a_t = \sum_{i=0}^t \delta^t  \ . 
\end{align}

Next, we relate $F_t$ to the dynamics.
\begin{theorem}\label{th:lb-general}
    For any $p$ such that 
    $F_t(p)<0$ it cannot be $p_t = p$.
\end{theorem}
\begin{proof}
    See Appendix~\ref{app:remaining_proofs}.
\end{proof}
Intuitively, this  theorem states that  if we show $F_t(p)<0$ for all $p<p^*$, then the minimum admission price is at least $p^*$.

\subsection{An illustrative example}\label{ex:lb-refined}

We next apply the result in Theorem~\ref{th:lb-general} to one of the simplest demand functions and show that, even in this case, price $p_{ser}$ is not a tight bound for the minimum admission price. 

\begin{proposition}[Lower bound]\label{prop: improvedlowerbound}
    For demand function $Q(p) = 1-p$ and $s\geq \frac{1}{2}$ the price at step $t\geq 1$ is at least $p^\star_t = \dfrac{1-{\delta}}{2(1 - \delta^t)}$ for any $\delta\in (0,1)$. Therefore, the minimum admission price is at least $p^\star= \frac{1-\delta}{2}$.
\end{proposition}

\begin{proof}
    See Appendix~\ref{app:remaining_proofs}.
\end{proof}

    \begin{remark}\label{rem:no-tight-line}
        According to the previous result, for $\delta = \frac{1}{2}$ the dynamics never go below $\frac{1}{4}$. That is, the minimum admission price is at least $\frac{1}{4}$ and transactions paying less than this price are never admitted. We note experimentally that this bound may not be tight, as the dynamics for $\delta = \frac{1}{2}$ never passes value $\approx 0.363$ which is computed over 100 steps (see Figure~\ref{fig:line-delta-0.5} showing the dynamics for different values of $\delta$).

        \begin{figure}[t]
    \centering
    \includegraphics[width=\linewidth]{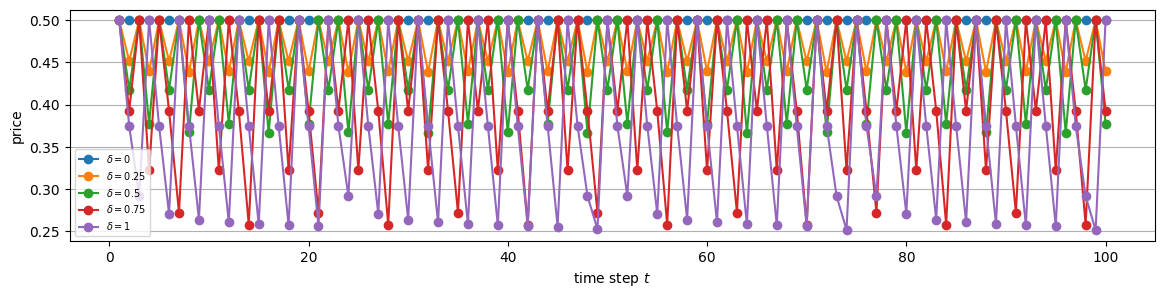}
    \caption[Caption without FN]{The dynamics for $Q(p)=1-p$ and $\delta \in [0,0.25,0.5,0.75,1]$ when supply $s=1$ (see Example~\ref{intro:example}). The smallest prices for each $\delta$ correspond to the respective minimum admission prices.\footnotemark}
    \label{fig:line-delta-0.5}
\end{figure}
\footnotetext{Note that this figure displays only 100 iterations and for a price dynamic to reach its minimum it may take more steps.}
    \end{remark}

    \begin{remark}
    The minimum admission price for $\delta=1$ goes to 0 for increasing $s$. Instead, for any $\delta < 1$ and any $s\geq \frac{1}{2}$, the minimum admission price is bounded away from 0, and is in fact at least $\frac{1-\delta}{2}$ as shown in the proposition above.
    \end{remark}

    \begin{remark}\label{re;lb-delta-small-question}
    For any arbitrarily small $\delta> 0$, Proposition~\ref{prop: improvedlowerbound} leaves open the possibility that  the minimum admission price is \emph{strictly smaller} than  $p_{mon}=\frac{1}{2}$. In the next section, we show that this is in general not the case, as there are other demand functions for which the minimum admission price \emph{equals} the monopolist price already for  $\delta>0$.
    \end{remark}

\subsection{The Admission Price Must Depend on $Q$}\label{sec: Qepsilon}
In this section, we consider the following class $\mathcal{Q}$ of daily demand functions:
\begin{equation}\label{piecewise_fct}
    Q_{\epsilon}(p) = \begin{cases}
        \frac{1}{2}+\epsilon - 2 \epsilon p, & 0\leq p\leq \frac{1}{2}\\
        1-p, & \frac{1}{2} \leq p\leq 1
    \end{cases}
\end{equation}
A function $Q_{\epsilon}$ of this class $\mathcal{Q}$ is depicted in Figure~\ref{fig:piecewise}. In particular, note that the slope of the function on the interval $[0,0.5]$ is depending on $\epsilon\geq 0$. On the remaining interval $[0.5,1]$ the function is just $1-p$.

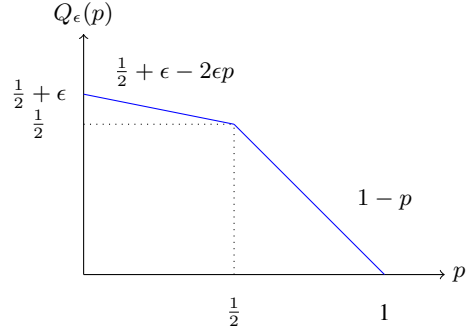
\begin{figure}[t]
    \centering
    \begin{tikzpicture}[scale=4]
  \draw[->] (0, 0) -- (1.2, 0) node[right] {$p$};
  \draw[->] (0, 0) -- (0, 0.8) node[above] {$Q_{\epsilon}(p)$};
  \draw[dotted] (0.5, 0) -- (0.5, 0.5);
  \draw[dotted] (0, 0.5) -- (0.5, 0.5);
  \node[xshift=-0.6cm] at (0,0.5) {$\frac{1}{2}$};
  \node[yshift=-0.5cm] at (0.5,0) {$\frac{1}{2}$};
  \node[xshift=-0.6cm] at (0,0.6) {$\frac{1}{2}+\epsilon$};
  \node[yshift=-0.5cm] at (1,0) {1};
  \node[yshift=-0.5cm] at (0.3,0.8) {$\frac{1}{2}+\epsilon - 2 \epsilon p$};
  \node[xshift=1cm] at (0.75,0.25) {$1-p$};
  \draw[domain=0:0.5, smooth, variable=\x, blue] plot ({\x}, {0.5+0.1- 2* 0.1*\x});
  \draw[domain=0.5:1, smooth, variable=\x, blue] plot ({\x}, {1-\x});
\end{tikzpicture}
    \caption{Daily demand function of the class $\mathcal{Q}$.}
    \label{fig:piecewise}
\end{figure}

A few observations are in place. 
\begin{description}
    \item [\textbf{Obs 1.}] Theorem~\ref{thm1} provides an upper bound on the minimum admission price if 
$\delta > \overline{\delta}_{ser}$.  According to \eqref{eq:def-pserbar-qserbar}, this condition is equivalent to   $\frac{q_{ser}-\overline{q}_{ser}}{1-\delta}>s$.\footnote{
For $\epsilon=\frac{1}{2}$ and $s=1$, the condition in \textbf{Obs 1} is to $\delta > 11/12 \approx 0.917$, $\frac{Q_{\epsilon}(p)-Q_{\epsilon}(p')}{1-\delta(\epsilon)}>s \iff \delta > 11/12 \approx 0.917.$
}
    \item [\textbf{Obs 2.}] For the type of demand function $Q_{\epsilon} \in \mathcal{Q}$, Theorem~\ref{thm1} may not apply, unless \textbf{Obs 1} is satisfied. This also depends on the value of $\epsilon$.
    \item [\textbf{Obs 3.}] The lower bound $p_{ser}$ is not tight (Remark~\ref{rem:no-tight-line} deals with $\epsilon=\frac{1}{2}$).
\end{description}

These observations naturally suggest to obtain \emph{lower bounds} on the minimum admission price for the class of functions above.

\begin{example}\label{ex:low-delta-Q-lb}
   Theorem~\ref{thm1} applies only for $\delta$ that are large enough ($\delta > \bar{\delta}_{ser}$). The necessary lower bound for $\delta$ for demand function of class $\mathcal{Q}$ is calculated below, that is, for the demand functions as in Equation~\eqref{piecewise_fct} and $s=1$, we have (by Definition~\ref{def:adapted-quantities})
       $
        p_{mon} = \frac{1}{2} = q_{mon}, p_{ser} = \frac{1}{4}, q_{ser} = \frac{1 + \epsilon}{2}, \overline{p}_{ser}=  \frac{1}{2(1 + \epsilon)} ,
        \overline{q}_{ser} = \frac{1}{2} + \frac{\epsilon^2}{1 + \epsilon},  
        \delta_{\min}(\overline{p}_{ser}) =   1 -\frac{\epsilon}{2} + \frac{\epsilon^2}{1 + \epsilon}, 
        \TM{p_{ser}}{\delta} = \frac{1}{4} + \frac{1-\delta}{\epsilon}$.
    Thus, for $Q(p)=1-p$, and $\epsilon=\frac{1}{2}$, we have 
    $\TM{p_{ser}}{\delta} = \frac{1}{4} + 2(1-\delta)$ for all $\delta > \bar{\delta}_{ser}=\delta_{\min}(\overline{p}_{ser}) =  11/12$. 
\end{example}

Next, we state three results for the class of demand functions $\mathcal{Q}$. First, for any $\delta$ we can find a strictly decreasing demand function such that the price dynamics are stuck at $p_{mon}$. Second, for all demand functions $Q_{\epsilon} \in \mathcal{Q}$ with $\epsilon < \frac{1}{2}$, we find a $\delta$ such that the same holds. Finally, we find a lower bound on the minimum admission price for the remaining case $\epsilon \geq \frac{1}{2}$.
\begin{theorem}\label{prop:Q-eps-fun} The following holds:
    \begin{enumerate}
        \item \label{item1} For every $\delta$ there exists a strictly decreasing $Q_{\epsilon} \in \mathcal{Q}$ s.t. the price dynamics stays at $p_{mon}$, i.e. 
        \begin{equation}\label{eq:stay}
            p_t = p_{mon}, \quad \text{for all $t$.}
        \end{equation}
        \item \label{item2} Conversely, for any strictly decreasing $Q_{\epsilon} \in \mathcal{Q}$ with $\epsilon < \frac{1}{2}$ we find a $\delta$ such that \eqref{eq:stay} holds.
        \item \label{item3} For every $Q_\epsilon \in \mathcal{Q}$ with $\epsilon \geq \frac{1}{2}$, the minimum admission price is at least $p^\star = \frac{1-\delta}{4\epsilon} < p_{mon}$.
    \end{enumerate} 
\end{theorem}

\begin{proof}
    See Appendix~\ref{app:remaining_proofs}.
\end{proof}

The first two items of the theorem above show that the dynamics ``collapse'' to the monopolist price for any $\delta > 0$ for some demand function, answering negatively the question raised in Remark~\ref{re;lb-delta-small-question}. The last item of the theorem generalizes the result in Proposition~\ref{prop: improvedlowerbound}, as the latter corresponds to $\epsilon = 1/2$. Note that the above result for $\epsilon\rightarrow 0$ is in line with Nisan's extension to constant demand functions (see Appendix~\ref{app:constant-func}).

\subsection{General demand function $Q$}
So far we have applied the tool from Theorem~\ref{th:lb-general} to specific demand functions. In this section, we instead derive a general bound. 
Let $Q$ be a general non-increasing demand function. We provide a lower bound for the minimum admission price that is higher than $p_{ser}$.

\begin{proposition}\label{prop:generalQ}
    Let $Q$ be a general non-increasing demand function such that $Q(0) \leq s+\delta (q_{mon}-s)$. Then, the price
    \begin{equation}
        p^\star := \frac{p_{mon}q_{mon}}{\frac{Q(0)}{1-\delta} + q_{mon} - \frac{q_{mon}}{1-\delta}}, 
    \end{equation}
    is a lower bound for the minimum admission price. In particular, $p^\star \geq p_{ser}$.
\end{proposition}

\begin{proof}
    See Appendix~\ref{app:remaining_proofs}.
\end{proof}

One example where the conditions in Proposition~\ref{prop:generalQ} are satisfied follows.
\begin{example}
    For $s = \frac{Q(0)}{1-\delta}$, the condition above ($Q(0) \leq s+\delta (q_{mon}-s)$) is satisfied and ${p}^\star = \frac{p_{mon}q_{mon}}{s-\frac{\delta}{1-\delta}q_{mon}}$ which is larger than $p_{ser}$.
\end{example}

\section{Conclusion and Future Work}
Our work makes a step toward understanding how a (strategic) user can make more informed decisions about transaction fees by leveraging their belief about others' patience level, contributing to the balance between cost savings and timely inclusion. Specifically, we examined the minimum transaction fee (admission price) users must pay to guarantee that their transaction is eventually included.

We analyzed price fluctuations under a monopolistic pricing mechanism with \emph{quasi-patient} users, whose interest in transaction inclusion diminishes over time due to costs. Our model includes impatient users \cite{Lavi2022} and patient users \cite{Nisan} as special cases ($\delta = 0$ and $\delta = 1$). We provide general bounds on price dynamics for $\delta \in (0, 1)$, highlighting key factors affecting dynamics and the minimum admission price. 

We first discuss the tightness of our bounds. 
Though Theorem~\ref{thm1} is a generic upper bound, i.e., it applies to any demand function, it is  essentially tight as shown by the next result (whose proof is in Appendix~\ref{app:tighness}).

\begin{theorem}\label{thm:tightness-main}
	For any arbitrarily small $\xi>0$, there is a demand function for which (i) the minimum admission price is strictly smaller than $p_{mon}$ for all $\delta>\bar{\delta}_{ser}$, where $\bar{\delta}_{ser}$ is defined as in Theorem~\ref{thm1}, and (ii) for all $\delta < \bar{\delta}_{ser} - \xi$ the minimum admission price equals to $p_{mon}$. 
\end{theorem}

Our work opens up interesting directions for future research. First, it might be interesting to consider restricted classes of demand functions of practical interest and provide tight bounds for \emph{these} functions. Note that Theorem~\ref{thm:tightness-main} may not apply, and our results suggest the need for some extension or variant of Theorems~\ref{thm1} and/or a stronger lower bound. Second, several extensions and variants are possible, leading to more complex dynamics. These include (i) various waiting cost functions, (ii) non-myopic monopolists maximizing revenue over multiple rounds, and (iii) strategic behavior of users and block producers. We believe some of our techniques can be adapted to handle more complex settings. In particular, dynamics where the level of patience correlates with the price could be analyzed by evaluating the counterpart of our pent-up ``growth'' term $a_t$ in \eqref{eq:at-bt-quantities} and relating it to some ``steepness'' condition of the demand curve, similar to the proof of Theorems~\ref{thm1} and \ref{th:lb-general}. Moreover, if the pent-up demand of certain complex dynamics can be ``sandwiched'' between two dynamics in our model, then our results could be turned into bounds for the dynamics of interest. We conjecture that more complex blockchain dynamics are qualitatively similar to ours and that this will help in obtaining useful parameter trade-offs and better design choices.

Finally, an open question remains regarding the comparison of social welfare in the case of quasi-patient users with the social welfare obtained in~\cite{Nisan}. On one hand, Theorem~\ref{thm:tightness-main} suggests that the dynamics collapses to the case $\delta=0$. On the other hand, it is unclear whether this occurs for demand functions where the social welfare for $\delta=0$ is arbitrarily bad relative to the optimum. Indeed, the proof of Theorem~\ref{thm:tightness-main}, uses demand curves with a ``flat region'' for which the social welfare for $\delta=0$ is within a constant factor of that for $\delta=1$ (the latter is the patient case in \cite{Nisan} for which the social welfare is at least $1/2$ of the optimum).

\begin{credits}
    \subsubsection{\ackname} We thank an anonymous reviewer for highlighting an issue in the proof of Theorem~6 in an earlier version of this work related to social welfare. While this observation raises an interesting point, the validity of the result remains an open question, as discussed above.
\end{credits}
\bibliographystyle{splncs04}
\bibliography{bibliography}

\newpage
\appendix
\section{Postponed Proofs}\label{app}
Section~\ref{app:thm1} provides the proof of Theorem~\ref{thm1} and Section~\ref{app:remaining_proofs} provides the proofs of the remaining results.

\subsection{Proof of Theorem~\ref{thm1}}\label{app:thm1}

This section provides lemmas (and its proofs) that will be used to prove Theorem~\ref{thm1}. 

\subsubsection{High level intuition.} Before presenting the proof in full detail, we first give a high level idea of the proof and how it differs from the one in \cite{Nisan} for the case $\delta = 1$. 
\begin{enumerate}
    \item In the quasi-patient users setting ($\delta < 1$), the ``accumulated'' pent-up demand is \emph{bounded} and cannot grow indefinitely as $t \rightarrow \infty$ (see Equations~\eqref{eq:D-rewritten} and \eqref{eq:at-bt-quantities}).
    \item Only for ``sufficiently large'' $\delta \in (0,1)$ does the ``accumulated'' pent-up demand suffice for a candidate price $p^* > p_{ser}$ to be better than the monopolist price $p_{mon}$ (otherwise the dynamics will not choose this price $p^*$). Similarly, for a given $\delta$, only ``sufficiently large'' prices $p^* > \TM{p_{ser}}{\delta}$ can be better than the monopolist price.
    \item The precise relation between the ``feasible'' $\delta$ and $p^*$ is given by the quantities in Definition~\ref{def:adapted-quantities}. The main technical contribution is to show that for such pairs (i) the dynamics must eventually take a price smaller than $p^*$ and (ii) the dynamics cannot ``rest'' on some $p<p_{mon}$ but must pass infinitely often through  $p_{mon}$ and followed by a sequence of decreasing prices repeating (i).
\end{enumerate}

\subsubsection{Actual proof in full detail.}
Throughout this section we use the quantity $a_t$ defined in Equation~\eqref{eq:at-bt-quantities}, which we rewrite here for convenience:
\begin{align}\label{eq:at-quantities-app}
    a_t = 1 + \delta+\cdots + \delta^{t-1} = \frac{1 - \delta^t}{1 - \delta} \ ,  && \delta \in [0,1)\ .
\end{align}
Note that some of the lemmas below follow the style of \cite{Nisan}, but are adapted to our setting to account for the impact of $\delta$ in the dynamics.

\begin{lemma}\label{le:Nisan-3}
For $p<p'$ and any $t$:
    $$D_t(p)-D_t(p') \leq a_t \cdot  (Q(p)-Q(p'))$$
    where $a_t$ is defined as in \eqref{eq:at-quantities-app}.
\end{lemma}
\begin{proof}
Fix any $t \geq 1$. We first show that the following holds:
    \begin{equation}\label{ineq1}
        Z_{t-1}(p)-Z_{t-1}(p') \leq D_{t-1}(p)- D_{t-1}(p') \ . 
    \end{equation} We distinguish the following cases:
    \begin{enumerate}
    \item ($p_{t-1}<p<p'$.) We have that $Z_{t-1}(p)=Z_{t-1}(p')=0$ and hence $D_t(p)-D_t(p') = [Q(p)-Q(p')]$
    \item ($p<p'<p_{t-1}$.) We have that $Z_{t-1}(p) = D_{t-1}(p) - q_{t-1}$ and $Z_{t-1}(p') = D_{t-1}(p') - q_{t-1}$. Together, $Z_{t-1}(p)-Z_{t-1}(p') = D_{t-1}(p)-D_{t-1}(p')$ and hence $D_t(p)-D_t(p') = \delta(D_{t-1}(p)-D_{t-1}(p')) + [Q(p)-Q(p')]$. 
    \item ($p<p_{t-1}<p'$.) We have that $Z_{t-1}(p) = D_{t-1}(p) - q_{t-1}$ and $Z_{t-1}(p') =0$. Also note that $D_{t-1}(p')\leq q_{t-1}= D_{t-1}(p_{t-1})$. Hence,
    $Z_{t-1}(p)-Z_{t-1}(p') = D_{t-1}(p) - q_{t-1} \leq D_{t-1}(p)- D_{t-1}(p')$. And therefore we have that, 
    \begin{equation}\label{ineq:D_t}
        D_t(p)-D_t(p') \leq \delta (D_{t-1}(p)- D_{t-1}(p')) + [Q(p)-Q(p')]
    \end{equation}
\end{enumerate}
Note that the inequality \eqref{ineq1} holds in all cases.
We next prove the lemma  by induction on $t$.
\begin{itemize}
    \item \emph{Base case} ($t=1$): $D_1(p)-D_1(p') = Q(p)- Q(p') \leq 1 [Q(p) - Q(p')]$
    \item \emph{Inductive step}: Assume the claim is true for $t-1$. 
    Then we have 
    \begin{align*}
        D_t(p)-D_t(p') &= \delta (Z_{t-1}(p)-Z_{t-1}(p')) + [Q(p) - Q(p')] \\ &\leq \delta (D_{t-1}(p)- D_{t-1}(p')) + [Q(p) - Q(p')]\\
        &\leq \delta ((1+ \delta + \delta^2 +\cdots+\delta^{t-2}) [Q(p)-Q(p')]) + [Q(p) - Q(p')]\\
        &= (1+ \delta + \delta^2 +\cdots+\delta^{t-1})[Q(p) - Q(p')]
    \end{align*}
    where we used Equation \eqref{ineq1} in the second line, and the second inequality holds by the three cases analysis above.
\end{itemize}
\end{proof}

\begin{lemma}
    For $p<p'$ and for all $T$ and $t>T$:
\begin{align}
    D_t(p)-D_t(p') \leq  a_{t - T} \cdot (Q(p) - Q(p')) + \delta^{t-T} \cdot (Z_{T}(p)-Z_{T}(p'))\ , 
\end{align}
where $a_t$ is defined as in \eqref{eq:at-quantities-app}.
\end{lemma}
\begin{proof}
    Similar to the proof above. By repeating Equation \eqref{ineq:D_t}, $(t-T-1)$-times and in another step we just replace $D_{T}(p)- D_{T}(p')$ with its definition, we end up with the expression:
\begin{equation}
    D_t(p)-D_t(p') \leq \delta^{t-T} (Z_{T}(p)-Z_{T}(p')) + (1 + \delta + \cdots + \delta^{t-T-1}) [Q(p) - Q(p')].
\end{equation}
By induction, it holds for all $t$:
\begin{itemize}
        \item $t=1$: $D_{1}(p)-D_{1}(p') = \delta (Z_{0}(p) - Z_{0}(p')) + [Q(p)- Q(p')] = [Q(p)- Q(p')]$
        \item Assume the claim is true for $t-1$
        \item $t-1 \rightarrow t$: 
            \begin{align*}
                &D_t(p)-D_t(p') \\&= \delta (Z_{t-1}(p)-Z_{t-1}(p')) + [Q(p) - Q(p')] \\ &\leq \delta (D_{t-1}(p)- D_{t-1}(p')) + [Q(p) - Q(p')]\\
                &\leq \delta (\delta^{t-T-1} (Z_{T}(p)-Z_{T}(p')) + (1 + \delta + \cdots + \delta^{t-T-2}) [Q(p) - Q(p')]) + [Q(p) - Q(p')]\\
                &= \delta^{t-T} (Z_{T}(p)-Z_{T}(p')) + (1 + \delta + \cdots + \delta^{t-T-1}) [Q(p) - Q(p')].
            \end{align*}
\end{itemize}
\end{proof}

\begin{lemma}\label{le:Nisan-4c}
    For all $T$ and $t>T$, if for all $t'$ such that $T<t'<t$ we also have that $p_{t'} \geq p'>p$ then in fact we have equality
\begin{equation}
    D_t(p)-D_t(p') =  a_{t-T} \cdot (Q(p) - Q(p')) + \delta^{t-T} (Z_{T}(p)-Z_{T}(p'))
\end{equation}
where $a_t$ is defined as in \eqref{eq:at-quantities-app}.
\end{lemma}
\begin{proof}
    For $p_{t-1}\geq p' > p$ we have that $Z_{t-1}(p) = D_{t-1}(p) - q_{t-1}$ and $Z_{t-1}(p') = D_{t-1}(p') - q_{t-1}$. Therefore, 
    $D_{t}(p)-D_{t}(p') = \delta (Z_{t-1}(p) - Z_{t-1}(p')) + [Q(p)- Q(p')]= \delta (D_{t-1}(p) - D_{t-1}(p')) + [Q(p)- Q(p')]$. By induction the claim holds.
    \begin{itemize}
        \item $t=1$: $D_{1}(p)-D_{1}(p') = \delta (Z_{0}(p) - Z_{0}(p')) + [Q(p)- Q(p')] = [Q(p)- Q(p')]$
        \item Assume the claim is true for $t-1$
        \item $t-1 \rightarrow t$: 
            \begin{align*}
                &D_t(p)-D_t(p') \\&= \delta (Z_{t-1}(p)-Z_{t-1}(p')) + [Q(p) - Q(p')] \\ &= \delta (D_{t-1}(p)- D_{t-1}(p')) + [Q(p) - Q(p')]\\
                &= \delta ((1 + \delta + \cdots + \delta^{t-T-2}) [Q(p) - Q(p')] + \delta^{t-T-1} (Z_{T}(p)-Z_{T}(p'))) + [Q(p) - Q(p')]\\
                &= (1 + \delta + \cdots + \delta^{t-T-1}) [Q(p) - Q(p')] + \delta^{t-T} (Z_{T}(p)-Z_{T}(p')).
            \end{align*}
    \end{itemize}
\end{proof}

\begin{lemma}\label{le:lemma4}
    For all $T$ such that $p_T\leq p<p'$ (or $T=0$) and all $t>T$, we have
\begin{equation}
    D_t(p)-D_t(p') \leq  a_{t-T} \cdot (Q(p) - Q(p'))
\end{equation}
where $a_t$ is defined as in \eqref{eq:at-quantities-app}.
Furthermore, for all $T$ such that $p_T\leq p<p'$ (or $T=0$), if for all $t'$ such that $T<t'<t$ we also have that $p_{t'}\geq p' > p $, then the equation above holds with equality, i.e. 
\begin{equation}
    D_t(p)-D_t(p') =  a_{t-T} \cdot (Q(p) - Q(p')) \ . 
\end{equation}
\end{lemma}
\begin{proof}
Apply the previous two lemmas and note that $Z_T(p) = Z_T(p')=0$ since $p_T\leq p<p'$.
\end{proof}

The proofs of Lemmas~\ref{lemma:lowerbound} and \ref{lem: upperbound} below are the same as in \cite{Nisan}, and we restate them here for the sake of completeness. 

\begin{lemma}\label{lemma:lowerbound}
For every $t$ it holds that $p_t\geq p_{ser}$.
\end{lemma}
\begin{proof}
    The maximum revenue that is achievable from a price $p$ is $p\cdot s$. For $p < p_{ser}$ we have that $p\cdot s< p_{ser}\cdot s=p_{mon} \cdot q_{mon}$, and the latter revenue can be achieved at any step using the monopolist price.
\end{proof}

\begin{lemma}\label{lem: upperbound}
    For every $t$ either $p_t=p_{mon}$ or $p_t < p_{t-1}$.
\end{lemma}
\begin{proof} 
    For $p \geq p_{t-1}$ we have that $D_t(p) = Q(p)$ so the maximal revenue obtained by possible $p \geq p_{t-1}$ is exactly the monopolist’s revenue that is obtained at $p_t = p_{mon}$ (we assume that ties in maximum revenue are broken consistently). So, unless $p_t = p_{mon}$, we must obtain the maximum revenue for some $p < p_{t-1}$.
\end{proof}

The next lemma provides a sufficient condition for the price to decrease.

\begin{lemma}\label{le:Nisan9}
    Assume that for some $p>p_{ser}$ we have that $D_t(p)\geq s$, then
    \begin{itemize}
        \item [(i)] $p_t<p_{t-1}$, and
        \item [(ii)] $Q(p_t)-Q(p_{t-1}) \geq (a_{t-1})^{-1} s \frac{(p-p_{ser})}{p_{mon}} $
    \end{itemize}
    where $a_t$ is defined as in \eqref{eq:at-quantities-app}.
    Furthermore, if for some $T < t$ we had $p_T \leq p_t$ then, 
    $Q(p_t)-Q(p_{t-1}) \geq (a_{t-1-T})^{-1} s \frac{(p-p_{ser})}{p_{mon}}$.
\end{lemma}

\begin{proof}
The proof follows closely the proof of Lemma~9 in \cite{Nisan}. We observe the following:
    We cannot have $p_t=p_{mon}$, as the revenue obtained from $p$ would be higher: $$p s> p_{ser}s = p_{mon}q_{mon}.$$ 
       As $p_t$ gives better revenue than $p$, i.e. we have that $$p_t q_t=p_t D_t(p_t) \geq p s = (p-p_{ser})s + p_{ser}s =  (p-p_{ser})s + p_{mon}q_{mon}.$$
        Separating the total demand at time $t$ to its two components we get
        \begin{align*}
            p_t D_t(p_t) &= \delta p_t Z_{t-1}(p_t) + p_t Q(p_t)\\
                         &\leq\footnotemark \delta p_t Z_{t-1}(p_t) + p_{mon}q_{mon}\\
                         &\leq \delta p_{mon} Z_{t-1}(p_t)+ p_{mon}q_{mon}.
        \end{align*}\footnotetext{In this step we use that $p_{mon}$ maximizes $pQ(p)$ and $q_{mon}=Q(p_{mon})$. Also we know that $p_t<p_{mon}$}
        \item Putting these together we get that 
        \begin{align}
             (p-p_{ser})s + p_{mon}q_{mon} \leq p_t D_t(p_t) \leq \delta p_{mon} Z_{t-1}(p_t)+ p_{mon}q_{mon},
        \end{align}
        that is
        \begin{align}
            (p-p_{ser})s  \leq \delta p_{mon} Z_{t-1}(p_t)\ .
        \end{align}
        Now observe that
        \begin{align*}
            Z_{t-1}(p_t) &= Z_{t-1}(p_t) - Z_{t-1}(p_{t-1})\\
            &\leq D_{t-1}(p_t) - D_{t-1}(p_{t-1}) \\ 
            &\leq (1+\delta+\ldots+\delta^{t-2}) [Q(p_t)-Q(p_{t-1})],
        \end{align*}
        where we used Lemma~\ref{le:Nisan-3} in the last step.
       Hence, it follows that
        \begin{align}
            (p-p_{ser})s  \leq p_{mon}(1+\delta+\ldots+\delta^{t-2}) [Q(p_t)-Q(p_{t-1})]
        \end{align}
        and therefore, 
        \begin{align}
            Q(p_t)-Q(p_{t-1}) \geq (1+\delta+\ldots+\delta^{t-2})^{-1} s \frac{(p-p_{ser})}{p_{mon}}.
        \end{align}
        The second part of the lemma is similar after applying the bound in Lemma~\ref{le:lemma4}:
        \begin{equation}
    D_t(p_t)-D_t(p_{t-1}) \leq  a_{t-T} \cdot (Q(p_t) - Q(p_{t-1})) \ . 
\end{equation}
This completes the proof.
\end{proof}

The remaining technical lemmas and results incorporate the following  condition (which is necessary in our setting): We consider pairs of prices such that 
\begin{align}\label{eq:extra-condition}
    \frac{1}{1-\delta}\cdot (Q(p)-Q(p'))> s && p < p' \ . 
    \tag{gap}
\end{align}
Intuitively speaking, \eqref{eq:extra-condition} says that the function $Q$ is ``sufficiently steep'' for the given $\delta$ and $s$. Next lemma then states that this is enough for the dynamics to eventually get below some price $p'$ as there is some smaller $p<p'$ with ``sufficiently high'' demand.

\begin{lemma}  \label{le:exceed-s}
    For every $p < p'$ satisfying \eqref{eq:extra-condition} there exists $\Delta_0$ s.t. for all $T$ and all $\Delta \geq \Delta_0$ we have that either (a) there exists $T\leq t \leq T+\Delta$ with $p_t<p'$ or (b) $D_{T+\Delta}(p)\geq s$.
\end{lemma}
\begin{proof}
    Let us observe that, if $p_t \geq p'$ for all $T\leq t \leq T+\Delta$, then using Lemma~\ref{le:Nisan-4c} we get
    \begin{align*}
        D_{T+\Delta}(p) & \geq D_{T}(p) - D_{T}(p') \\
    &\stackrel{Lem~\ref{le:Nisan-4c}}{=}  a_\Delta \cdot (Q(p) - Q(p')) + \delta^{\Delta} (Z_{T}(p)-Z_{T}(p'))\\
        &\geq a_\Delta \cdot (Q(p) - Q(p')) = (1 + \delta + \cdots + \delta^{\Delta-1}) \cdot (Q(p) - Q(p'))\ . 
    \end{align*}
    We next show that, for sufficiently large $\Delta$, the latter quantity must exceed $s$. 
    That is, we can find $\Delta_0$ such that $(1 + \delta + \cdots + \delta^{\Delta_0-1}) \cdot (Q(p) - Q(p'))= s$:
    \begin{align*}
        \frac{1-\delta^{\Delta_0}}{1-\delta} = \frac{s}{Q(p) - Q(p')} && \Leftrightarrow &&
        \delta^{\Delta_0} &= 1-\frac{s (1-\delta)}{Q(p) - Q(p')} \\ 
        && \Leftrightarrow && 
        \Delta_0 &= \ln\left((1-\delta)-\frac{s (1-\delta)}{Q(p) - Q(p')}\right) \\ && \Leftrightarrow && 
        \Delta_0 &= \ln\left(1-\delta \right) + \ln\left(1-\frac{s}{Q(p) - Q(p')}\right) \ .
    \end{align*}
    This completes the proof. 
\end{proof}

The next lemma quantifies the ``minimal $\delta$'' such that \eqref{eq:extra-condition}  holds for a given price $p^*$, and therefore the dynamics eventually goes below $p^*$. 

\begin{lemma}\label{le:Nisan11}
For every $p^\star>p_{ser}$ such that
    \begin{align}\label{eq:p-gap-condition}
        \delta > \delta_{\min}(p^\star) :=  1-\frac{Q(p_{ser})-Q(p^\star)}{s} 
    \end{align}
    the following holds.
    There exists $\Delta$ such that for every $T$ there exists some $T < t \leq T + \Delta$ with $p_t \leq p^\star$.
\end{lemma}
\begin{proof}
    By contradiction, assume $T=0$ or $T$ such that  $p_t > p^\star$ for all $t\in[T,T+\Delta]$, for all $\Delta$. Let 
    \begin{align}
         \label{eq:lem-11:p-def}
         p := \frac{p^\star +p_{ser}}{2}
    \end{align}
    so $p_{ser} < p < p^\star$ and $p^\star - p_{ser} = 2(p-p_{ser})$. Next observe that \eqref{eq:p-gap-condition} says that the condition \eqref{eq:extra-condition} required in Lemma~\ref{le:exceed-s} holds for $p = p_{ser}$ and $p'=p^\star$. Hence, there exists $\Delta_0$ after which $D_t(p) \geq s$ for all $t\geq T+\Delta_0$ until the first time that $p_t\leq p^\star$. Fix $\Delta>\Delta_0$~\footnote{It may happen that such a $\Delta$ does not exist. In particular, if $D_t(p) \geq s$ only for $t=T+\Delta_0$ and for $t=T+\Delta_0 +1$ we may have $p_t\leq p^\star$.} such that $D_t(p) \geq s$ for all $t \in [T+\Delta_0, T+\Delta]$.
    By Lemma~\ref{le:Nisan9}  we get a sequence of decreasing prices
    \begin{align}
        p_{T+\Delta_0} > p_{T+\Delta_0+1} > \cdots > p_{T+\Delta} 
    \end{align}
    such that, for all $t\in [T+\Delta_0+1, T + \Delta]$,  we have 
    \begin{align}
        Q(p_{t}) - Q(p_{t-1}) \geq \frac{1}{a_{t-1}}
        \cdot \frac{p - p_{ser}}{p_{mon}}\cdot s \ . 
    \end{align}
    Hence, using $a_t = \frac{1-\delta^t}{1-\delta}$, we have 
    \begin{align*}
        Q(p_{T+\Delta}) - Q(p_{T+\Delta_0}) &\geq  
        \sum_{t=T+\Delta_0+1}^{T+\Delta} \frac{1}{a_{t-1}}\cdot \frac{p - p_{ser}}{p_{mon}}\cdot s\\&=
        \left( \frac{1- \delta}{1 - \delta^{\Delta_0}} + \cdots  + \frac{1-\delta}{1 - \delta^{\Delta -1}}\right)\cdot \frac{p - p_{ser}}{p_{mon}}\cdot s \\ &\geq  
        \frac{1-\delta}{1 - \delta^{\Delta_0}} \cdot (\Delta - \Delta_0) \cdot \frac{p - p_{ser}}{p_{mon}}\cdot s \ . 
    \end{align*}
    Next observe that  $Q(p_{T+\Delta}) - Q(p_{T+\Delta_0}) \leq Q(p_{ser})-Q(p_{mon})$, since all prices are between $p_{ser}$  and $p_{mon}$, and $Q$ is decreasing. We therefore get a contradiction if this inequality holds:
    $$Q(p_{ser})-Q(p_{mon}) < \frac{1-\delta}{1 - \delta^{\Delta_0}} \cdot (\Delta - \Delta_0) \cdot \frac{p - p_{ser}}{p_{mon}}\cdot s \ .$$ 
    This condition is true whenever $$\Delta> \Delta_0 + Q(p_{ser})-Q(p_{mon}) \frac{p_{mon}(1-\delta^{\Delta_0})}{s(p-p_{ser})(1-\delta)}$$ and thus we have the contradiction for such $\Delta$. We conclude that we must have $p_t \leq p^\star$.
\end{proof}

We next provide a technical lemma relating $\delta$ and values of $p^\star$ for which the dynamics must take the monopolist price infinitely often. 
\begin{lemma}\label{lem: nisan13}
    Suppose there exist two values $p^\star>p_{ser}$ and $q^\star>q_{ser}$ satisfying \eqref{eq:p-gap-condition} and the following inequality:\footnote{Note that $q^\star \neq Q(p^\star)$ and in particular $q^\star>q_{ser}=Q(p_{ser})>Q(p^\star)$.}
    \begin{align}\label{eq:pq-pmonqmon}
        p^\star q^\star <  p_{mon}q_{mon}=p_{ser} \cdot s\ . 
    \end{align}
    Then, there exist infinitely many $t$ such that $p_t = p_{mon}$.
\end{lemma}
\begin{proof}By contradiction, assume there is a last time $\ell\geq 1$ such that the dynamics takes the monopolist price. We then observe the following:
\begin{enumerate}
    \item We have a sequence of decreasing prices (Lemma~\ref{lem: upperbound}) 
    \begin{align}
        p_{\ell} > p_{\ell+1} > \cdots
    \end{align}
    satisfying $p_t \geq p_{ser}$ for all $t \geq \ell$ (Lemma~\ref{lemma:lowerbound}). 
    \item By Lemma~\ref{le:Nisan11} there is some $T_0\geq \ell$ such that, for all $t\geq T_0$, we have 
    \begin{align}
        p_t < p^\star \ ,  && p_tq_t \geq p_{mon}q_{mon}=p_{ser}\cdot s && \overset{\eqref{eq:pq-pmonqmon}}{\Rightarrow} && q_t > q^\star 
    \label{eq:le:nisan13:implication}
    \end{align}
    where the second inequality holds because the dynamics prefers  $p_t$  to $p_{mon}$. 
    \item The monotonicity of $Q$ and $p_t \geq p_{ser}$ imply $q_t\leq q_{ser}<q^\star$, thus contradicting   \eqref{eq:le:nisan13:implication}.  
\end{enumerate}
This completes the proof.
\end{proof}

We next turn our attention to the existence of $p^\star$ and $q^\star$.

\begin{lemma}\label{le:delta-vs-price}
For every continuous $Q$ such that the equilibrium revenue is less than the monopolist revenue, the following hold:
\begin{enumerate}
    \item \label{itm:le:delta-vs-price-1}
    Quantities in \eqref{eq:def-pserbar-qserbar} satisfy $\overline{p}_{ser}>p_{ser}$ and  $\overline{q}_{ser}< q_{ser}$.
    \item \label{itm:le:delta-vs-price-2} For every $p^\star>p_{ser}$ satisfying $p^\star<\overline{p}_{ser}$ condition \eqref{eq:pq-pmonqmon} in Lemma~\ref{lem: nisan13} holds for some $q^\star>q_{ser}$. 
    \item \label{itm:le:delta-vs-price-3} For every $\delta>\overline{\delta}_{ser}$, there exists $p^\star>p_{ser}$ satisfying $p^\star<\overline{p}_{ser}$ such that condition \eqref{eq:p-gap-condition} in Lemma~\ref{le:Nisan11} holds. In particular, this holds true for any $\TM{p_{ser}}{\delta} < p^\star<\overline{p}_{ser}$. 
\end{enumerate}
\end{lemma}
\begin{proof}
We distinguish the three parts:
\begin{enumerate}
    \item As shown in \cite{Nisan}, if the equilibrium revenue is smaller than the monopolist revenue, then $q_{ser}<s$. Indeed, the equilibrium  is given by the price $p_{eq}$ such that $q_{eq}:=Q(p_{eq})=s$, and $\texttt{REV}_{eq}=p_{eq}s<\texttt{REV}_{mon} = p_{mon}q_{mon}=p_{ser}s$, thus implying $p_{eq} < p_{ser}$. Since $Q$ is decreasing, the previous inequality yields $s=Q(p_{eq}) > Q(p_{ser})=q_{ser}$. Since $q_{ser}<s$, we have  $\bar{p}_{ser}=\frac{p_{ser}s}{q_{ser}}>p_{ser}$   and, again using monotonicity of $Q$, also $\bar{q}_{ser}=Q(\bar{p}_{ser}) < Q(p_{ser}) = q_{ser}$. 
    \item Since $p^\star < \bar{p}_{ser}=\frac{p_{ser}s}{q_{ser}}$, we have that $p^\star q_{ser}< \bar{p}_{ser} q_{ser} = p_{ser}s$. Therefore, for a sufficiently small $\rho>1$, we have that  $p^\star q_{ser}\cdot \rho < p_{ser}s$. Hence, by taking $q^\star :=q_{ser} \rho > q_{ser}$ the desired condition is satisfied: $$p^\star q^\star   = p^\star q_{ser}\cdot \rho < p_{ser}s = p_{mon}q_{mon} \ .$$
    \item  
    We first show that $\TM{p_{ser}}{\delta}< \overline{p}_{ser}$.  Observe that 
    \begin{align}
         \delta_{\min}(\TM{p_{ser}}{\delta}) = & 1-\frac{q_{ser}-Q(\TM{p_{ser}}{\delta})}{s} \\  \stackrel{\eqref{eq:def-pserdelta}}{=} &  1-\frac{q_{ser}-(q_{ser} - (1-\delta)\cdot s)}{s} = \delta \label{eq:deltamin-pdelta-delpta}\\
        > &  \overline{\delta}_{ser} = 1-\frac{q_{ser}-Q(\overline{p}_{ser})}{s} \ . 
    \end{align}
    Therefore,  $Q(\TM{p_{ser}}{\delta})>Q(\overline{p}_{ser})$ and, by monotonicity of $Q$, we have $\TM{p_{ser}}{\delta}< \overline{p}_{ser}$. For any $p^\star$ such that $\TM{p_{ser}}{\delta} < p^\star<\overline{p}_{ser}$ we have
    $Q(p^\star) < Q(\TM{p_{ser}}{\delta})$ and therefore 
\begin{align}
        \delta_{\min}(p^\star) =  1-\frac{q_{ser}-Q(p^\star)}{s} <   1-\frac{q_{ser} -Q(\TM{p_{ser}}{\delta})}{s} \stackrel{\eqref{eq:deltamin-pdelta-delpta}}{=} \delta \ . 
    \end{align}

 \end{enumerate}
This completes the proof. 
\end{proof}

We are now in a position to prove  Theorem~\ref{thm1}, which we restate here for convenience. 
\begin{T1}
   For any strictly decreasing demand function $Q$ and supply $s$ the following holds:
    \begin{enumerate}
        \item The minimum admission price is at least $p_{ser}$, and thus transactions paying less than this price will never be included. In particular, the dynamics stay always between $p_{ser}$ and $p_{mon}$, that is,  prices $p_t$ satisfy $p_{ser}\leq p_t \leq p_{mon}$ for all $t\geq 1$. 
        Moreover, at each step $t$, the prices either decrease ($p_t < p_{t-1}$) or they jump up to the monopolist price ($p_t=p_{mon}$). 
        \item
        For every  $\delta>\overline{\delta}_{ser}$, the minimum admission price is at most $\TM{p_{ser}}{\delta}$ defined by \eqref{eq:def-pserdelta} which satisfies $p_{ser}< \TM{p_{ser}}{\delta} < \overline{p}_{ser}$.
        Moreover, the dynamics  pass through the monopolist price $p_{mon}$ infinitely often. 
        \item Every price  larger than $p_{ser}$ is an admission price for a sufficiently large $\delta$. 
        That is, for every $p^\star >p_{ser}$, there exists $\delta_{min}(p^\star)<1$ such that $p^\star$ is an admission price for every $\delta > \delta_{min}(p^\star)$. 
        Moreover, the dynamics  pass through the monopolist price $p_{mon}$ infinitely often. 
    \end{enumerate}
   Therefore, transactions paying at least $p_{mon}$ are immediately included, and this is tight to guarantee immediate inclusion (as there are infinitely many steps for which paying less will delay admission to a later step). 
\end{T1}
\begin{proof}
    We distinguish the three parts:
    \begin{enumerate}
        \item The bounds on the prices are due to Lemmas~\ref{lemma:lowerbound} and \ref{lem: upperbound}. The condition on the price changes is simply the observation that for $p\geq p_{t-1}$ we have $D_t(p) = Q(p)$ and thus the dynamics either take the monopolist price ($p_t=p_{mon}$) or take a smaller price ($p_t<p_{t-1}$). 
        \item Consider any $\TM{p_{ser}}{\delta}$ such that $\TM{p_{ser}}{\delta} < p^\star<\overline{p}_{ser}$. Lemma~\ref{le:Nisan11} together with Lemma~\ref{le:delta-vs-price} (Item~\ref{itm:le:delta-vs-price-3}) imply  that there exists $\Delta$ such that for every $T$ there exists some $T < t \leq T + \Delta$ with $p_t \leq p^\star$. Lemma~\ref{le:delta-vs-price} (Item~\ref{itm:le:delta-vs-price-2}) states that the conditions in Lemma~\ref{lem: nisan13} hold for  $p^\star>p_{ser}$. The latter implies that the dynamics takes the monopolist price infinitely often. 
        
        \item The first part follows directly from Lemma~\ref{le:Nisan11}
    and from the fact that $p^\star > p_{ser}$ implies $Q(p^\star) <  Q(p_{ser})$ and thus $\delta_{\min}(p^\star)<1$. The second part follows from Lemma~\ref{lem: nisan13} and Item~\ref{itm:le:delta-vs-price-2} of Lemma~\ref{le:delta-vs-price}.
    \end{enumerate}
    This completes the proof.
\end{proof}

\subsection{Remaining Proofs}\label{app:remaining_proofs}
In this section, we provide the remaining proofs.

\begin{proof}[Proof of Theorem~\ref{th:lb-general}]
We first rewrite Equation~\eqref{revenue-diff} using Equation~\eqref{eq:D-rewritten} as follows:
\begin{align}
    f_t(p) = & p\cdot D_t(p) - p_{mon}q_{mon} = p \cdot (a_t Q(p) - b_t) - p_{mon}q_{mon} \ ,
    \end{align}
    where $a_t$ and $b_t$ are defined in \eqref{eq:at-bt-quantities}. 
For any $\overline{a}_t \geq a_t$ and any $\underline{b}_t \leq b_t$ we can obtain an upper  bound on $f_t(\cdot)$:
\begin{align}\label{eq:min-p-function-ub}
    f_t(p) \leq \overline{f}_t(p) := p \cdot (\overline{a}_t Q(p) - \underline{b}_t) - p_{mon}q_{mon} \ .
\end{align}
For any $p$ such that $\overline{f}_t(p)<0$, we obviously have $f_t(p)< 0$, which implies that it cannot be $p_t = p$. Indeed, definition of  $f_t(p)< 0$ implies $p\cdot D_t(p) < p_{mon}q_{mon}$ and hence $p$ will not be taken.

Next observe that, since $p_t \leq p_{mon}$, the monotonicity of $Q$ implies $q_t \geq q_{mon}$, thus implying 
\begin{align}\label{eq:min-p-function-b-term-simple}
    b_t \geq (a_t-1) q_{mon} =: \underline{b}_t  \ . 
\end{align}
Finally, observe that  for $\underline{b}_t$ we have $\overline{f}_t(p)$ is exactly $F_t$ in \eqref{eq:lb-improved-F}. This completes the proof.  
\end{proof}

\begin{proof}[Proof of Proposition~\ref{prop: improvedlowerbound}]
Though Theorem~\ref{prop:Q-eps-fun} implies this result, we  next provide the proof of this proposition to illustrate the technique on a slightly simpler case. 

Theorem~\ref{th:lb-general} yields a lower bound for the smallest $p$ of the dynamics. Since $p_{mon}=\frac{1}{2}=q_{mon}$ we have
    \begin{align}
       F_t(p) 
       =&  \left(\dfrac{\left(1-{\delta}^t\right)\left(1-p\right)}{1-{\delta}}-\dfrac{\frac{1-{\delta}^t}{1-{\delta}}-1}{2}\right)p-\dfrac{1}{4}
    \end{align}
    which has two roots. One root is at $\frac{1}{2} (=p_{mon})$ and the other one is  
    \begin{align}
        p^\star_t = \dfrac{1-{\delta}}{2(1 - \delta^t)} \ . 
    \end{align}
    We can make the following observations:
    \begin{enumerate}
        \item For $t\rightarrow \infty$ and fixed $\delta\in (0,1)$ we have $p^\star_t \rightarrow p^\star:= \frac{1-\delta}{2}$ meaning that the dynamics cannot go below this minimum price $p^\star$. 
        \item For $\delta \rightarrow 1$  and for fixed $t\geq 2$ we have $p^\star_t \rightarrow \frac{1}{2t}$ meaning that the dynamics cannot be below $\frac{1}{2t}$ at step $t$.   
    \end{enumerate}
        This completes the proof.
    \end{proof}

\begin{proof}[Proof of Theorem~\ref{prop:Q-eps-fun}]
Let us assume $s=1$ (this is wlog since for any $s$ we can simply consider a ``rescaled'' demand function $Q_\epsilon(p)/s$). We apply Theorem~\ref{th:lb-general} to  $Q_\epsilon \in \mathcal{Q}$ and consider the corresponding function 
\begin{align}    F_t(p) = p\cdot\left(\dfrac{\left(1-{\delta}^t\right)\left(-2{\epsilon}p+{\epsilon}+\frac{1}{2}\right)}{1-{\delta}}-\dfrac{\frac{1-{\delta}^t}{1-{\delta}}-1}{2}\right)-\dfrac{1}{4}, 
    \end{align}
which has two roots. One root is at $\frac{1}{2}$($=p_{mon}$) and the second root is \begin{equation}\label{root2}
    p^\star_t:=\dfrac{{\delta}-1}{4{\delta}^t{\epsilon}-4{\epsilon}}.
\end{equation}
Function $F_t(p)$ is concave as its second derivative is $-\frac{\left(4{\delta}^{t} - 4\right) {\epsilon}}{{\delta} - 1}<0$. Hence, it assumes positive values only for $p$ in the interval between these two roots.
Next observe  that
\begin{align}
    p^\star_t \overset{t \to \infty}{\rightarrow} \frac{1-\delta}{4\epsilon} =: p^\star.
\end{align}
Furthermore, we have that $p^\star_t<p^\star$. Therefore, for $t \rightarrow \infty$ and  $\epsilon < \frac{1-\delta}{2}$ this root $p^\star_t$ is bigger than the monopolist price and this implies that the dynamics never go below $p_{mon}$. Similarly, for a fixed $\epsilon$ (i.e. a given $Q_{\epsilon} \in \mathcal{Q}$) we can find the $\delta$ that satisfies the latter condition, the dynamics never go below $p_{mon}$.

Furthermore, if $\epsilon \geq \frac{1-\delta}{2(1-\delta^t)}$ for all $t$, then the root in \eqref{root2} is below $\frac{1}{2}$, which gives us a lower bound for the price dynamics.
Note that, if $\epsilon \geq \frac{1-\delta}{2(1-\delta^1)} = \frac{1}{2}$ holds, then $\epsilon \geq \frac{1-\delta}{2(1-\delta^t)}$ for all $t$.
In particular, for a fixed $\epsilon$ we can find two intervals for $\delta$ such that we end up with a lower bound less than $p_{mon}$ on one interval and a lower bound equal to $p_{mon}$ on the other interval.
\end{proof}

\begin{proof}[Proof of Proposition~\ref{prop:generalQ}]
Similar to earlier calculations, we consider the function in Equation~\eqref{eq:lb-improved-F}.
\begin{align}
       F_t(p) = & p   (a_tQ(p) - (a_t-1) q_{mon}) - p_{mon}q_{mon}, 
\end{align}
and
\begin{align}
       \overline{\overline{f}}_t(p) := & p   (a_t Q(0) - (a_t-1) q_{mon}) - p_{mon}q_{mon},
\end{align}
where $\overline{\overline{f}}_t(p)$ uses the upper bound on the demand function.
Note that $f_t(p) \leq F_t(p) \leq \overline{\overline{f}}_t(p)$ where $f_t(p)$ is as in Equation~\eqref{revenue-diff}.
Being linear and increasing in $p$, note that $\overline{\overline{f}}_t(p)$ has a root at \begin{equation}
\overline{\overline{p}}_t = \frac{p_{mon}q_{mon}}{a_t Q(0) - (a_t-1) q_{mon}} \overset{t \to \infty}{\to} \frac{p_{mon}q_{mon}}{\frac{Q(0)}{1-\delta} + q_{mon} - \frac{q_{mon}}{1-\delta}} =: p^\star.    
\end{equation}
Hence this $p^\star$ is a new lower bound if $Q(0) \leq s+\delta (q_{mon}-s)$ since then $p^\star\geq p_{ser}$. The proof is completed.
\end{proof}

\section{Tightness of the Bounds}\label{app:tighness}

In this section, we discuss the tightness of our upper and lower bounds (Theorem~\ref{thm1} and Theorem~\ref{prop:Q-eps-fun}). We observe the following:
\begin{itemize}
	\item Calculations in Example~\ref{ex:low-delta-Q-lb} together with Theorem~\ref{thm1} state that the minimum admission price can be upper bounded for $
	\delta \geq \bar{\delta}_{ser} = \delta_{\min}(\overline{p}_{ser}) =   1 -\frac{\epsilon}{2} + \frac{\epsilon^2}{1 + \epsilon}$ and in this case the minimum admission price is at most $
	\TM{p_{ser}}{\delta} = \frac{1}{4} + \frac{1-\delta}{\epsilon} < p_{mon}$.
	\item The proof of Theorem~\ref{prop:Q-eps-fun} shows that, for  $\epsilon < \frac{1-\delta}{2}$, the dynamics get stuck at $p_{mon}$. Equivalently, the minimum admission price is $p_{mon}$ for   $ \delta < \bar{\delta}_{ser}^-:=1 - 2\epsilon$.
\end{itemize}
Putting things together we obtain the following result:
\begin{theorem}\label{thm:tightness}
	For any positive $\epsilon<1/2$ there exist two values $\bar{\delta}_{ser}^-$ and $\bar{\delta}_{ser}$, with  $\bar{\delta}_{ser}^-< \bar{\delta}_{ser}$, such that for  demand function $Q_\epsilon \in \mathcal{Q}$ the following holds. For all $\delta > \bar{\delta}_{ser}$ the minimum admission price is strictly smaller than the monopolist price.  Moreover, for all $\delta < \bar{\delta}_{ser}^-$ , the minimum admission price equals the monopolist price. Finally, the difference between these two thresholds vanishes as $\epsilon \rightarrow 0$, 
	\begin{align}
		\bar{\delta}_{ser} - \bar{\delta}_{ser}^- =  (1 -\frac{\epsilon}{2} + \frac{\epsilon^2}{1 + \epsilon})  - (1 - 2\epsilon)  = \frac{3\epsilon}{2} + \frac{\epsilon^2}{1 + \epsilon} < 2\epsilon\ . 
	\end{align}
\end{theorem}
\begin{proof}[Proof of Theorem~\ref{thm:tightness-main}] Theorem~\ref{thm:tightness-main} follows from Theorem~\ref{thm:tightness} with $\epsilon= \frac{\xi}{2}$.
 \end{proof}
 
\section{Discussion about Lower Bound and Constant Functions}\label{app:constant-func}
In this section, we discuss the case when the demand function is not strictly decreasing, but rather constant for some interval of prices. In particular, we discuss the case $\epsilon=0$ from Section~\ref{sec: Qepsilon}, and show that the  result for $\epsilon\rightarrow 0$ is in line with Nisan's extension to constant demand functions.

We observe the following:
\begin{itemize}
    \item Taking $\epsilon=0$ in our setting, yields the  \emph{piecewise constant} demand function $Q_0 \in \mathcal{Q}$ shown in Figure~\ref{fig:piecewise_eps=0}. For any $s\geq \frac{1}{2}$, the price dynamics stays at $p_{mon}=\frac{1}{2}$ which can be verified by observing that for any $p<p_{mon}$ and any $t\geq 1$ we have $D_t(p) = \frac{1}{2}$. Therefore, to maximize revenue, the monopolist will always choose $p_t=p_{mon}$, and thus the minimum admission price equals $p_{mon}=\frac{1}{2}$. Note that  this happens for all $\delta$, including $\delta=1$. 
    \item For patient users,  Nisan \cite{Nisan} observed that for weakly decreasing (piecewise constant) demand functions, the lower bound $p_{ser}$  is no longer tight and needs to be adjusted to some $\bar{p}_{ser}$ defined as follows. 
    This new lower bound is the largest $p$ such that $$Q_0(p_{ser}) = Q_0(p_{mon}q_{mon}/s).$$ In this example, $p_{mon}q_{mon}=\frac{1}{4}$ and   $Q_0(p_{mon}q_{mon}/s)=\frac{1}{2}$ for $s\geq \frac{1}{2}$. Clearly, the largest $p$ such that $Q_0(p)=\frac{1}{2}$ is $p=\frac{1}{2}(=:\bar{p}_{ser})$ which is the new lower bound. That is, also for $\delta=1$ the minimum admission price equals  $p_{mon}=\frac{1}{2}$.
\end{itemize}

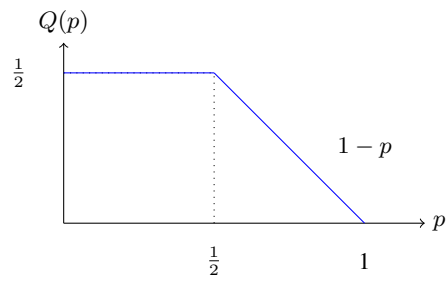
\begin{figure}[h]
    \centering
    \begin{tikzpicture}[scale=4]
  \draw[->] (0, 0) -- (1.2, 0) node[right] {$p$};
  \draw[->] (0, 0) -- (0, 0.6) node[above] {$Q(p)$};
  \draw[dotted] (0.5, 0) -- (0.5, 0.5);
  \draw[dotted] (0, 0.5) -- (0.5, 0.5);
  \node[xshift=-0.6cm] at (0,0.5) {$\frac{1}{2}$};
  \node[yshift=-0.5cm] at (0.5,0) {$\frac{1}{2}$};
  \node[yshift=-0.5cm] at (1,0) {1};
  \node[xshift=1cm] at (0.75,0.25) {$1-p$};
  \draw[domain=0:0.5, smooth, variable=\x, blue] plot ({\x}, {0.5});
  \draw[domain=0.5:1, smooth, variable=\x, blue] plot ({\x}, {1-\x});
\end{tikzpicture}
    \caption{Daily demand function of class $\mathcal{Q}_\epsilon$ for $\epsilon=0$.}
    \label{fig:piecewise_eps=0}
\end{figure}

\end{document}